\documentclass[11pt]{article}

\usepackage[margin=1.05in]{geometry}
\usepackage{amsmath,amssymb,amsthm,mathtools}
\usepackage{booktabs}
\usepackage{enumitem}
\usepackage{microtype}
\usepackage{xcolor}
\usepackage{aliascnt}
\usepackage[colorlinks=true,linkcolor=blue!55!black,citecolor=blue!55!black,urlcolor=blue!60!black]{hyperref}
\usepackage[nameinlink,capitalise,noabbrev]{cleveref}

\newtheorem{theorem}{Theorem}[section]
\newaliascnt{lemma}{theorem}
\newtheorem{lemma}[lemma]{Lemma}
\aliascntresetthe{lemma}
\newaliascnt{proposition}{theorem}
\newtheorem{proposition}[proposition]{Proposition}
\aliascntresetthe{proposition}
\newaliascnt{corollary}{theorem}
\newtheorem{corollary}[corollary]{Corollary}
\aliascntresetthe{corollary}
\theoremstyle{definition}
\newaliascnt{definition}{theorem}
\newtheorem{definition}[definition]{Definition}
\aliascntresetthe{definition}
\newaliascnt{remark}{theorem}

\aliascntresetthe{remark}

\crefname{theorem}{Theorem}{Theorems}
\crefname{lemma}{Lemma}{Lemmas}
\crefname{proposition}{Proposition}{Propositions}
\crefname{corollary}{Corollary}{Corollaries}
\crefname{definition}{Definition}{Definitions}
\crefname{remark}{Remark}{Remarks}

\newcommand{\FB}{\operatorname{FB}}
\newcommand{\SO}{\operatorname{SO}}
\newcommand{\BO}{\operatorname{BO}}
\newcommand{\RO}{\operatorname{RO}}
\newcommand{\GFT}{\operatorname{GFT}}
\newcommand{\E}{\mathbb E}
\newcommand{\Prb}{\mathbb P}

\newcommand{\Rtwo}{\mathcal R_2}
\newcommand{\BetaFn}{\mathrm B}

\hypersetup{
  pdftitle={Tighter Bounds for the Random-Offerer Mechanism in Bilateral Trade},
  pdfauthor={Sunghyeon Jo},
  pdfsubject={Approximation guarantees for bilateral trade},
  pdfkeywords={bilateral trade, gains from trade, random offerer, mechanism design}
}

\title{Tighter Bounds for the Random-Offerer Mechanism\\in Bilateral Trade}
\author{Sunghyeon Jo}
\date{July 2026}

\begin{document}
\maketitle

\begin{abstract}
The random-offerer mechanism for bilateral trade selects the seller or the
buyer uniformly and lets the selected agent make a profit-maximizing
take-it-or-leave-it offer.  Let $\rho_{\rm RO}$ be the infimum, over
independent value distributions, of the mechanism's gains from trade divided
by first-best gains from trade.  We prove
\[
  \frac1\pi\le \rho_{\rm RO}<0.460242308085529.
\]
For the lower bound, we improve the previous guarantee from approximately
$0.317844$ to $1/\pi\approx 0.318310$.  The proof uses a parameterized
Lagrangian bound for pointwise-monotone allocations.  At multiplier one, this bound has
coefficient $2/\pi$, and the Lagrangian separates into two terms controlled by
the optimal seller-offering and buyer-offering profits.  For the upper bound,
we construct an explicit family consisting of a truncated equal-revenue buyer
and a seller distribution with a tilted power-law lower tail and a
constant-virtual-cost segment.  The family satisfies
$\FB/\RO>2.17276852308451$, improving the previous explicit ratio $2.0749$;
rigorous interval arithmetic certifies the numerical inequality.
\end{abstract}

\section{Introduction}

Bilateral trade is the basic setting in which private information can prevent
all gains from trade from being realized.  A seller owns one indivisible item
and values it at $S$, while a buyer values it at $B$.  The values are drawn
independently from commonly known distributions, and each agent observes only
its own realization.  The first-best allocation trades whenever $B>S$ and
generates expected gains from trade
\[
   \FB:=\E[(B-S)_+].
\]
Myerson and Satterthwaite~\cite{MS83} showed that no mechanism can attain this
benchmark for all priors while satisfying Bayesian incentive compatibility,
individual rationality, and budget balance.

The random-offerer mechanism is a simple way to conduct trade under these
constraints.  It selects the seller or the buyer with equal probability and
lets the selected agent make a take-it-or-leave-it offer after observing its
value.  When the seller is selected, a seller of type $s$ chooses a price $p$
maximizing $(p-s)\Prb[B\ge p]$; when the buyer is selected, a buyer of type
$b$ chooses a price $p$ maximizing $(b-p)\Prb[S\le p]$.  Let $\SO$ and $\BO$
denote the gains from trade in the two cases and set
\[
   \RO:=\frac{\SO+\BO}{2}.
\]
We study the worst-case fraction of first-best gains from trade,
\[
   \rho_{\rm RO}:=
   \inf_{F_S,F_B:\,\FB>0}\frac{\RO}{\FB}.
\]

Deng, Mao, Sivan, and Wang~\cite{DMSW22} proved the first constant lower bound
on $\rho_{\rm RO}$.  Fei~\cite{Fei22} improved the guarantee to
$1/T=0.317844432899372\ldots$, where $T>1$ is the larger solution of
$T=2+\log T$; Hartline and Wang~\cite{HW25} later gave a geometric proof of
the same bound.  In the other direction, Babaioff, Dobzinski, and
Kupfer~\cite{BDK21} constructed instances on which random offerer obtains less
than one half of first best.  Cai, Gupta, Li, and Mehta~\cite{CGLM26} recently
obtained an explicit instance with $\FB/\RO\ge 2.0749$.  Thus the previously
known range was
\[
  0.317844432899372\ldots
  \le \rho_{\rm RO}
  \le \frac1{2.0749}
  =0.481950937394573\ldots.
\]

We tighten both sides of this range:
\begin{equation}\label{eq:main-interval}
  \frac1\pi\le \rho_{\rm RO}<0.460242308085529.
\end{equation}
The lower bound holds for every pair of independent priors on a common compact
interval.

\begin{theorem}\label{thm:pi-main}
For every pair of independent value distributions on a common compact interval,
\[
   \RO\ge \frac1\pi\,\FB.
\]
\end{theorem}

The upper bound is witnessed by an explicit family of distributions.

\begin{theorem}\label{thm:upper-main}
There exist independent value distributions supported on $[0,1]$, with a
continuous seller distribution and a buyer distribution having a continuous
part and one atom, such that
\[
   \frac{\FB}{\RO}>2.17276852308451.
\]
Consequently, $\rho_{\rm RO}<0.460242308085529$.
\end{theorem}

The exact value of $\rho_{\rm RO}$ remains open.  The lower-bound proof
yields a family of inequalities indexed by the Lagrange multiplier, and the
choice $\alpha=1$ is the point at which the Lagrangian matches the two offerer
profit problems.  The upper-bound construction gives an explicit analytic
family whose performance can be evaluated rigorously.  Together, the two
results substantially narrow the known range of $\rho_{\rm RO}$.

We next describe the two arguments.  Recent work of Liu, Qin, Ren, and
Wang~\cite{LQRW26} proves that second-best bilateral trade obtains at least one
half of first-best gains from trade.  Their proof uses a Lagrangian reduction
and a greedy construction of a pointwise-monotone allocation on a lattice.  We
keep the Lagrange multiplier $\alpha$ as a free parameter.  For smooth priors,
let $\Lambda(x)$ be the budget-balance functional and define
\[
   V_\alpha:=\sup_x\{\GFT(x)+\alpha\Lambda(x)\},
\]
where the supremum is over pointwise-monotone allocations.  We prove, for
every $\alpha>0$,
\begin{equation}\label{eq:C-beta-intro}
  V_\alpha\ge \frac{1}{C(\alpha)}\,\FB,
  \qquad
  C(\alpha):=
  \frac{1}{1+\alpha}
  \frac{\Gamma\!\left(\frac{1}{1+\alpha}\right)^2}
       {\Gamma\!\left(\frac{2}{1+\alpha}\right)}.
\end{equation}
At $\alpha=1$,
\[
   C(1)=\frac12\BetaFn\!\left(\frac12,\frac12\right)=\frac\pi2.
\]
The same choice of multiplier yields
\[
  \GFT(x)+\Lambda(x)
  =\E[(\phi_B(B)-S)x(S,B)]
   +\E[(B-\psi_S(S))x(S,B)],
\]
where $\phi_B$ is the buyer's virtual value and $\psi_S$ is the seller's
virtual cost.  A nondecreasing row of a pointwise-monotone allocation is a
mixture of seller posted-price rules, and a nonincreasing column is a mixture
of buyer posted-price rules.  The two terms are therefore bounded by the
optimal seller-offering and buyer-offering profits, $\Pi_S$ and $\Pi_B$.
Since gains from trade dominate proposer profit, we obtain
\[
   \frac{2}{\pi}\FB
   \le V_1
   \le \Pi_S+\Pi_B
   \le \SO+\BO
   =2\RO.
\]

The proof of \eqref{eq:C-beta-intro} begins with the consecutive-lattice
analysis of Liu et al.  On the lattice, we prove the greedy inequality for every $\beta$ satisfying
$\beta C(\alpha)\le1$.  We then use quantile ironing to remove the seller
regularity assumption without decreasing the optimized Lagrangian value.  A
regularity-preserving lattice approximation transfers the discrete bound to
smooth distributions, and a final approximation argument gives the stated
result for arbitrary Borel priors on a compact interval.

For the upper bound, the buyer has a truncated equal-revenue distribution.
This choice makes every positive seller type optimally post the price $1$.  The
seller CDF follows a tilted power law on a lower interval, has constant virtual
cost on a middle interval, and equals one above a cutoff.  The lower interval
makes the buyer's optimal offer follow a prescribed increasing map.  The
middle interval makes all larger buyer types offer the same boundary price;
seller types in that interval contribute to first best and to the
seller-offering mechanism but do not trade when the buyer makes the offer.
After sending the buyer truncation and seller perturbation to zero in sequence,
the ratio $\FB/\RO$ reduces to three one-dimensional integrals.  An explicit
choice of five rational parameters gives the limiting value
\[
   2.172768523084520237\ldots.
\]
The incentive properties of the construction are proved analytically, and
interval arithmetic supplies rigorous enclosures for the three integrals.

\section{Model and preliminaries}\label{sec:model}

The seller's value $S$ and buyer's value $B$ are independent Borel random
variables.  Their distributions are common knowledge.  We normalize their
compact support to $[0,1]$; a common positive affine change preserves every
gains-from-trade ratio in the paper.
Write $F_S$ and $F_B$ for their CDFs.  Whenever a density or point mass is
present, write it as $f_S$ or $f_B$, with the argument indicating which
interpretation applies.

\begin{definition}[Gains from trade and first best]\label{def:gft}
An allocation rule $x:[0,1]^2\to[0,1]$ gives the probability of trade at each
reported pair $(s,b)$.  Its expected gains from trade are
\[
   \GFT(x):=\E[(B-S)x(S,B)].
\]
The first-best rule is $x^{\rm FB}(s,b)=\mathbf 1\{b>s\}$, and
\[
   \FB:=\E[(B-S)_+].
\]
\end{definition}

We next define the seller- and buyer-offering mechanisms.  A proposer may post
a price in $[0,1]$ or choose no trade.  The responder accepts at equality.  On
the relevant range ($p\ge s$ for the seller and $p\le b$ for the buyer), each
profit objective is upper semicontinuous and therefore attains its maximum.
When several prices are optimal, choose the smallest one.  The lower-bound
argument is independent of this tie-breaking rule.

\begin{definition}[Seller-offering mechanism]\label{def:SO}
After observing $s$, the seller chooses a price $p\in[0,1]$ maximizing
\[
   (p-s)\Prb[B\ge p].
\]
The buyer accepts when $B\ge p$.  The mechanism's gains from trade are denoted
$\SO$.  The seller's ex-ante profit is
\[
   \Pi_S:=\E_S\!\left[\max\!\left\{0,\sup_{p\in[0,1]}
             (p-S)\Prb[B\ge p]\right\}\right],
\]
where a price with zero trade is available.
\end{definition}

\begin{definition}[Buyer-offering mechanism]\label{def:BO}
After observing $b$, the buyer chooses a price $p\in[0,1]$ maximizing
\[
   (b-p)\Prb[S\le p].
\]
The seller accepts when $S\le p$.  The mechanism's gains from trade are denoted
$\BO$.  The buyer's ex-ante profit is
\[
   \Pi_B:=\E_B\!\left[\max\!\left\{0,\sup_{p\in[0,1]}
             (B-p)\Prb[S\le p]\right\}\right].
\]
\end{definition}

\begin{definition}[Random-offerer mechanism and worst-case fraction]\label{def:RO}
The random-offerer mechanism selects each offering mechanism with probability one half:
\[
   \RO:=\frac{\SO+\BO}{2}.
\]
Its worst-case first-best fraction is
\[
   \rho_{\rm RO}:=\inf_{F_S,F_B:\,\FB>0}\frac{\RO}{\FB}.
\]
\end{definition}

The responder receives nonnegative utility whenever trade occurs.  This gives a
basic comparison used later.

\begin{lemma}[Gains from trade dominate proposer profit]\label{lem:gft-profit}
The offering mechanisms satisfy
\[
   \SO\ge\Pi_S,
   \qquad
   \BO\ge\Pi_B.
\]
\end{lemma}

\begin{proof}
At a seller price $p$, every accepted trade satisfies
$B-S=(p-S)+(B-p)\ge p-S$.  At a buyer offer $p$, every accepted trade satisfies
$B-S=(B-p)+(p-S)\ge B-p$.  Taking expectations proves both inequalities.
\end{proof}

\subsection{Monotone allocations and smooth virtual values}

\begin{definition}[Pointwise monotonicity]\label{def:monotonicity}
Let $\Rtwo$ be the set of measurable allocations $x:\mathbb R^2\to[0,1]$ such
that $x(s,b)$ is nondecreasing in $b$ for every $s$ and nonincreasing in $s$
for every $b$.
\end{definition}

For the fixed-multiplier argument, first suppose that both distributions have
strictly positive continuously differentiable densities on their support
intervals.  Their ordinary virtual value and virtual cost are
\begin{equation}\label{eq:ordinary-virtual}
 \phi_B(b):=b-\frac{1-F_B(b)}{f_B(b)},
 \qquad
 \psi_S(s):=s+\frac{F_S(s)}{f_S(s)}.
\end{equation}

\begin{definition}[Budget-balance functional]\label{def:discriminant}
For $x\in\Rtwo$, define
\begin{equation}\label{eq:Lambda}
 \Lambda(x):=\E[(\phi_B(B)-\psi_S(S))x(S,B)].
\end{equation}
\end{definition}

The envelope formula identifies $\Lambda(x)$ as the budget-balance functional:
$\Lambda(x)\ge0$ is the expected-payment condition
for weak budget balance under envelope-normalized transfers; see Myerson and
Satterthwaite~\cite{MS83}.  For
$\alpha\ge0$, define
\begin{align}
 \phi_B^\alpha(b)&:=(1+\alpha)b
     -\alpha\frac{1-F_B(b)}{f_B(b)},\notag\\
 \psi_S^\alpha(s)&:=(1+\alpha)s
     +\alpha\frac{F_S(s)}{f_S(s)}.
 \label{eq:alpha-virtual}
\end{align}
Then
\begin{equation}\label{eq:Lalpha-virtual}
 \GFT(x)+\alpha\Lambda(x)
 =\E[(\phi_B^\alpha(B)-\psi_S^\alpha(S))x(S,B)].
\end{equation}

\begin{definition}[Fixed-multiplier value]\label{def:Valpha}
For smooth priors and $\alpha\ge0$, set
\[
 V_\alpha:=\sup_{x\in\Rtwo}
 \{\GFT(x)+\alpha\Lambda(x)\}.
\]
\end{definition}

\begin{theorem}[Fixed-multiplier bound]\label{thm:fixed-alpha}
For every $\alpha>0$ and every pair of independent value distributions with
strictly positive continuously differentiable densities on a common compact
interval,
\[
 V_\alpha\ge\beta(\alpha)\,\FB.
\]
In particular, $C(1)=\pi/2$ and $V_1\ge(2/\pi)\FB$.
\end{theorem}

\section{A fixed-multiplier bound}\label{sec:fixed}

For $\alpha>0$, set
\begin{equation}\label{eq:aq}
 q:=\frac1{1+\alpha},
 \qquad
 a:=\frac{\alpha}{1+\alpha}=1-q.
\end{equation}
The constant in \cref{eq:C-beta-intro} is
\begin{equation}\label{eq:C-beta}
 C(\alpha)=q\BetaFn(q,q)
 =q\frac{\Gamma(q)^2}{\Gamma(2q)},
 \qquad
 \beta(\alpha)=\frac1{q\BetaFn(q,q)}.
\end{equation}

\subsection{The consecutive-lattice bound}

We first work on the unit lattice.  Let the seller have positive mass
$f_S(s)$ at every $s\in\{0,1,\ldots,W\}$ and put $F_S(-1)=0$.  Its lattice
$\alpha$-virtual cost is
\begin{equation}\label{eq:lattice-virtual}
 \psi_S^\alpha(s):=(1+\alpha)s
   +\alpha\frac{F_S(s-1)}{f_S(s)}.
\end{equation}
The seller is \emph{$\alpha$-weakly regular} when these costs are
nondecreasing.  Buyer values range over the nonnegative integers; the local
formulation below allows zero buyer masses and therefore covers every finite
buyer distribution.

On the unit lattice, write the fixed-multiplier objective directly in terms of
probability masses:
\begin{align}
 L_\alpha(x;F_S,F_B)
 &:=(1+\alpha)\sum_{s,b}f_S(s)f_B(b)(b-s)x(s,b)\notag\\
 &\quad-\alpha\sum_{s,b}f_S(s)(1-F_B(b))x(s,b)
 -\alpha\sum_{s,b}f_B(b)F_S(s-1)x(s,b).
 \label{eq:lattice-L}
\end{align}
It equals $\GFT+\alpha\Lambda$ when all masses are positive and is also defined
when some buyer masses vanish.

As in Liu, Qin, Ren, and Wang~\cite{LQRW26}, an allocation is defined on the
whole lattice and is pointwise monotone.  Entries with $b\le s$ can be set to
zero without lowering \cref{eq:lattice-L}.  Indeed, the coefficient of such an
entry is
\[
 f_S(s)\bigl((1+\alpha)f_B(b)(b-s)
       -\alpha(1-F_B(b))\bigr)
 -\alpha f_B(b)F_S(s-1)\le0.
\]
Fix a finite buyer truncation $\{0,\ldots,M\}$.  Let $\Delta_M$ be the buyer
probability simplex and let $\Rtwo^{W,M}$ be the finite pointwise-monotone
allocation polytope.  For
\[
 g(x,F_B):=L_\alpha(x;F_S,F_B)-\beta\FB(F_S,F_B),
\]
bilinearity and compactness allow Sion's minimax theorem~\cite{Sion58} to be
applied in the explicit form
\begin{equation}\label{eq:sion-explicit}
 \inf_{F_B\in\Delta_M}\sup_{x\in\Rtwo^{W,M}}g(x,F_B)
 =\sup_{x\in\Rtwo^{W,M}}\inf_{F_B\in\Delta_M}g(x,F_B).
\end{equation}
Thus it suffices to construct one allocation satisfying
\begin{equation}\label{eq:fixed-maximin}
 L_\alpha(x;F_S,F_B)-\beta\FB\ge0
 \quad\text{for every buyer distribution on $\{0,\ldots,M\}$}.
\end{equation}
After setting $x(s,b)=0$ for $b\le s$, the objective difference is
$g(x,F_B)=\sum_{b=0}^M f_B(b)H(b;x)$, where
\begin{equation}\label{eq:H-local}
 H(b;x):=
 \sum_{0\le s<b} f_S(s)
 \left[
 (b-s)((1+\alpha)x(s,b)-\beta)
 -\alpha\!\sum_{s<s'<b}x(s',b)
 -\alpha\!\sum_{s<b'<b}x(s,b')
 \right].
\end{equation}
The sums over intermediate lattice points are the discrete envelope rents.
Because a linear function on $\Delta_M$ attains its minimum at a vertex, the
inner infimum in \cref{eq:sion-explicit} equals
$\min_{0\le b\le M}H(b;x)$.  Thus \cref{eq:fixed-maximin} is equivalent to
$H(b;x)\ge0$ for every $b\le M$.  The construction and its no-failure proof
are uniform in $M$, so the same inequalities hold for every finite lattice
buyer distribution.

The greedy construction of Liu et al.~\cite[Theorem~4.6]{LQRW26} processes
buyer values increasingly, copies the preceding column, and raises the current
column from low seller types upward until $H(b;x)=0$.  Their
maximin--greedy equivalence is stated for arbitrary $\beta\ge0$.

\begin{proposition}[Parameterized greedy bound]\label{prop:capacity}
Fix $\alpha>0$ and a full-support $\alpha$-weakly regular seller distribution
on a consecutive unit lattice.  If $0<\beta C(\alpha)\le1$, then the greedy
construction satisfies every inequality \cref{eq:H-local}.
\end{proposition}

\begin{proof}
\Cref{app:parameter-audit} verifies the four steps of Liu et
al.~\cite[Proposition~5.4--Theorem~5.10]{LQRW26} with symbolic $\beta$:
threshold-one feasibility, threshold-stationary truncation, nonstationary
truncation, and the minimal-counterexample argument.
\end{proof}

\begin{corollary}[Regular lattice bound]\label{cor:lattice-bound}
For a full-support $\alpha$-weakly regular lattice seller and every
finite-support lattice buyer distribution,
\[
 \sup_{x\in\Rtwo}L_\alpha(x;F_S,F_B)
 \ge\beta(\alpha)\FB.
\]
\end{corollary}

\begin{proof}
Apply \cref{prop:capacity} with $\beta=1/C(\alpha)$ and use the
maximin--greedy equivalence.
\end{proof}

\subsection{Quantile ironing at a fixed multiplier}

Let $c(u)=F_S^{-1}(u)$ be the seller's quantile function, $0<u<1$, and write
\begin{equation}\label{eq:quantile-virtual}
 \psi_\alpha(u):=(1+\alpha)c(u)+\alpha u c'(u),
 \qquad
 \Phi(u):=\int_0^u\psi_\alpha(v)\,dv.
\end{equation}
The identity
\begin{equation}\label{eq:Phi-M}
 \Phi(u)=M(u)+\alpha u c(u),
 \qquad M(u):=\int_0^u c(v)\,dv,
\end{equation}
follows by integration by parts.

For a bounded seller coefficient $\chi:(0,1)\to\mathbb R$ and a smooth
buyer prior, let $\mathcal X$ be the set of measurable
$X:(0,1)\times\mathbb R\to[0,1]$ that are nonincreasing in the first
coordinate and nondecreasing in the second.  Define the quantile objective
and its value by
\begin{align}
 \mathcal L_\alpha(\chi;X)
 &:=\int_0^1\E_B\!\left[
       \bigl(\phi_B^\alpha(B)-\chi(u)\bigr)X(u,B)\right]du,
 \label{eq:quantile-objective}\\
 \mathcal V_\alpha(\chi,F_B)
 &:=\sup_{X\in\mathcal X}\mathcal L_\alpha(\chi;X).
 \label{eq:quantile-value}
\end{align}
When $c$ is strictly increasing and $\chi=\psi_\alpha$, the change of
variables $u=F_S(s)$ gives
$\mathcal V_\alpha(\psi_\alpha,F_B)=V_\alpha$.

Let $\overline\Phi$ be the greatest convex minorant of $\Phi$, and let
$\overline\psi$ be the right-continuous version of its nondecreasing slope.
The zero function is a convex minorant of the nonnegative curve $\Phi$, so
$0\le\overline\Phi\le\Phi$ and $\overline\Phi(0)=\Phi(0)=0$.  The minorant
also contacts $\Phi$ at $1$: otherwise its endpoint value could be raised
slightly without violating either convexity or the minorant property.  Since
$\Phi(u)=o(u)$ as $u\downarrow0$, convexity gives
$\overline\psi(0+)=0$.  Define
\begin{equation}\label{eq:continuous-ironed-c}
 \overline c(u):=
 \frac1\alpha u^{-(1+\alpha)/\alpha}
 \int_0^u v^{1/\alpha}\overline\psi(v)\,dv.
\end{equation}
Because a convex nonnegative function starting at zero has nonnegative
right derivative, $\overline\psi\ge0$.
Then
\begin{equation}\label{eq:continuous-ironed-ode}
 (1+\alpha)\overline c(u)+\alpha u\overline c'(u)
 =\overline\psi(u).
\end{equation}
More precisely, $(1+\alpha)\overline c(u)$ is a weighted average of
$\overline\psi(v)$ over $0<v<u$.  Since $\overline\psi$ is nondecreasing,
$(1+\alpha)\overline c(u)\le\overline\psi(u)$.  Equation
\eqref{eq:continuous-ironed-ode} then gives $\overline c'(u)\ge0$ almost
everywhere, so $\overline c$ is a seller quantile function after choosing its
left-continuous version.

\begin{lemma}[Fixed-multiplier quantile ironing]\label{lem:fixed-ironing}
For every smooth buyer prior,
\[
 \mathcal V_\alpha(\psi_\alpha,F_B)
 =\mathcal V_\alpha(\overline\psi,F_B),
 \qquad
 \FB(\overline c,F_B)\ge\FB(c,F_B).
\]
\end{lemma}

\begin{proof}
For $X\in\mathcal X$, let $y(u)=\E_B[X(u,B)]$, which is nonincreasing.  Put
$D=\Phi-\overline\Phi\ge0$.  Both curves agree at $0$ and $1$.  The
buyer term is unchanged when only the seller coefficient is ironed, and
Stieltjes integration by parts gives
\begin{align}
 \mathcal L_\alpha(\overline\psi;X)
   -\mathcal L_\alpha(\psi_\alpha;X)
 &=\int_0^1(\psi_\alpha-\overline\psi)y\,du\notag\\
 &=-\int_0^1D\,dy\ge0.
 \label{eq:ironing-one-way}
\end{align}
Let $\mathcal I$ be the countable collection of connected components of
$\{u:D(u)>0\}$.  On every $I=(\ell,r)\in\mathcal I$, the greatest convex
minorant is the chord joining the two contact points, and hence
$\overline\psi$ is constant.  Define $AX$ by replacing, for every $b$, the
rows on $I$ by their Lebesgue average over $I$, and leave all other rows
unchanged.  For a nonincreasing function, the average on an interval lies
between its one-sided endpoint values.  Thus these replacements preserve
monotonicity in $u$; averaging also preserves monotonicity in $b$.  They leave
$\int_0^1X(u,b)\,du$ unchanged for every $b$.  Therefore
\[
 \mathcal L_\alpha(\overline\psi;AX)
 =\mathcal L_\alpha(\overline\psi;X).
\]
Moreover, $AX$ is constant on each $I$ and
\[
 \int_I(\psi_\alpha-\overline\psi)\,du
 =D(r)-D(\ell)=0.
\]
On the complement of the ironing intervals, $D=0$ and hence
$D'=\psi_\alpha-\overline\psi=0$ almost everywhere.  It follows that
$\mathcal L_\alpha(\psi_\alpha;AX)
=\mathcal L_\alpha(\overline\psi;AX)$.  Consequently
$\mathcal V_\alpha(\overline\psi,F_B)
\le\mathcal V_\alpha(\psi_\alpha,F_B)$, while
\cref{eq:ironing-one-way} gives the reverse inequality.

For first best, define $\overline M(u)=\int_0^u\overline c(v)\,dv$.  The pairs
$(M,\Phi)$ and $(\overline M,\overline\Phi)$ satisfy
$M+\alpha uM'=\Phi$ and
$\overline M+\alpha u\overline M'=\overline\Phi$.  The solution with zero
initial value is
\[
 M(u)=\frac1\alpha u^{-1/\alpha}
       \int_0^u v^{1/\alpha-1}\Phi(v)\,dv.
\]
Thus $\overline\Phi\le\Phi$ implies $\overline M\le M$.  For every buyer
value $b$,
\[
 \int_0^1(b-c(u))_+\,du
 =\max_{0\le u\le1}\{bu-M(u)\}.
\]
Replacing $M$ by $\overline M$ weakly raises the maximum.  Averaging over $B$
proves the claim.
\end{proof}

\subsection{A regularity-preserving lattice limit}

The next lemma transfers the consecutive-lattice bound to the quantile
model while preserving regularity.

\begin{lemma}[Regular lattice approximation]\label{lem:regular-lattice-limit}
Fix $\alpha>0$.  For a bounded nonnegative nondecreasing function
$\psi:(0,1)\to\mathbb R$ with $\psi(0+)=0$, define
\begin{equation}\label{eq:T-alpha}
 (T_\alpha\psi)(u):=
 \frac1\alpha u^{-(1+\alpha)/\alpha}
 \int_0^u v^{1/\alpha}\psi(v)\,dv.
\end{equation}
Let $c=T_\alpha\psi$.  For every buyer prior having a strictly positive
continuously differentiable density on a compact interval beginning at zero,
\[
 \mathcal V_\alpha(\psi,F_B)
 \ge\beta(\alpha)\FB(c,F_B).
\]
\end{lemma}

\begin{proof}
We begin with a seller having a positive continuously differentiable density
$f$ on $[0,R]$ and a continuous, strictly increasing value-space
$\alpha$-virtual cost
\[
 \widehat\psi(s):=(1+\alpha)s+\alpha\frac{F(s)}{f(s)}.
\]
Write $c_F:=F^{-1}$ for this seller's quantile function.
For $W\ge1$, put $h=R/W$ and $s_i=ih$.  With
$q_i=f(s_i)/F(s_i)$ for $i\ge1$, define
\begin{equation}\label{eq:lattice-Euler-recursion}
 G_W=1,
 \qquad
 G_{i-1}:=\frac{G_i}{1+hq_i},
 \qquad
 g_i:=G_i-G_{i-1},
 \qquad
 g_0:=G_0.
\end{equation}
Set $G_{-1}:=0$.
The $g_i$ are positive and sum to one.  The original-scale lattice virtual
cost at $s_i$ is
\[
 \widehat\psi_{h,i}:=(1+\alpha)s_i
        +\alpha h\frac{G_{i-1}}{g_i}.
\]
For $i\ge1$, \cref{eq:lattice-Euler-recursion} gives
$G_{i-1}/g_i=F(s_i)/(h f(s_i))$, and hence
\begin{equation}\label{eq:exact-lattice-cost}
 \widehat\psi_{h,i}=\widehat\psi(s_i).
\end{equation}
The equality also holds at $i=0$, where both sides are zero.  Thus the seller
has full support and is $\alpha$-weakly regular.  Dividing all values by $h$
produces the consecutive unit-lattice seller in
\cref{cor:lattice-bound}; both the fixed-multiplier objective and first best
scale by $1/h$.

We next prove the convergence needed for the optimized objective.  Iterating
\cref{eq:lattice-Euler-recursion} gives
\begin{equation}\label{eq:lattice-product}
 G_i=\prod_{j=i+1}^{W}\left(1+h\frac{f(s_j)}{F(s_j)}\right)^{-1}.
\end{equation}
For every $\delta>0$, $f/F$ is bounded and continuously differentiable on
$[\delta,R]$.  Uniformly over $s_i\ge\delta$,
\begin{align*}
 \log G_i
 &=-\sum_{j=i+1}^{W}\log\left(1+h\frac{f(s_j)}{F(s_j)}\right)\\
 &=-\sum_{j=i+1}^{W}h\frac{f(s_j)}{F(s_j)}+O(h)
 \longrightarrow-\int_{s_i}^{R}\frac{f(t)}{F(t)}\,dt
 =\log F(s_i).
\end{align*}
Hence $G_i/F(s_i)\to1$ uniformly away from zero.  Moreover,
\begin{equation}\label{eq:lattice-density-ratio}
 \frac{g_i}{h}
 =\frac{G_i}{F(s_i)}
   \frac{f(s_i)}{1+h f(s_i)/F(s_i)}
 \longrightarrow f(s_i)
\end{equation}
uniformly on the same region.  The lattice mass below $\delta$ converges to at
most $F(\delta)$.  More precisely, if $i_h=\max\{i:s_i\le\delta\}$, then
$G_{i_h}\to F(\delta)$ at every continuity point of $F$; here every point is a
continuity point.  Thus the mass of the cells meeting $[0,\delta]$ is
$F(\delta)+o(1)$, which vanishes after first taking $h\downarrow0$ and then
$\delta\downarrow0$.

Let $I_i=[s_i,s_i+h)$, extend continuous seller densities by zero beyond
$R$, and spread each lattice weight uniformly over its cell:
\begin{equation}\label{eq:spread-seller}
 f_{S,h}(s):=\sum_{i=0}^{W}\frac{g_i}{h}\mathbf1_{I_i}(s),
 \qquad
 a_{S,h}(s):=\sum_{i=0}^{W}
       \frac{g_i\widehat\psi_{h,i}}{h}\mathbf1_{I_i}(s).
\end{equation}
Put
\[
 a_S(s):=f(s)\widehat\psi(s)
        =(1+\alpha)s f(s)+\alpha F(s)
\]
on $[0,R]$, with zero extension outside.  The uniform convergence in
\cref{eq:lattice-density-ratio}, the vanishing initial mass, and the $O(h)$
mass in the cell extending beyond $R$ imply
\begin{equation}\label{eq:seller-L1}
 \|f_{S,h}-f\|_1\longrightarrow0,
 \qquad
 \|a_{S,h}-a_S\|_1\longrightarrow0.
\end{equation}

Discretize the buyer upward on the same grid: $B_h=h\lceil B/h\rceil$ for
$B>0$, and set the zero column to zero.  Write $b_j=jh$,
$p_{j,h}=\Prb[B_h=b_j]$, and
$\tau_j=\Prb[B_h>b_j]=\Prb[B>b_j]$.  Zero masses are allowed.  The signed
buyer virtual coefficient at report $b_j$ is
\begin{equation}\label{eq:buyer-signed-coefficient}
 A^B_{j,h}:=(1+\alpha)b_jp_{j,h}-\alpha h\tau_j.
\end{equation}
For $J_j=((j-1)h,jh]$, spread these weights as
\begin{equation}\label{eq:spread-buyer}
 f_{B,h}(b):=\sum_{j\ge1}\frac{p_{j,h}}h\mathbf1_{J_j}(b),
 \qquad
 a_{B,h}(b):=\sum_{j\ge1}\frac{A^B_{j,h}}h\mathbf1_{J_j}(b).
\end{equation}
Only finitely many summands are nonzero.  If
$\tau(b)=\Prb[B>b]$ and
\[
 a_B(b):=(1+\alpha)b f_B(b)-\alpha\tau(b),
\]
then the first function in \cref{eq:spread-buyer} is the usual density
histogram, while on $J_j$ the second is
$(1+\alpha)b_jp_{j,h}/h-\alpha\tau(b_j)$.  Smoothness of the density and
Lipschitz continuity of the tail therefore give
\begin{equation}\label{eq:buyer-L1}
 \|f_{B,h}-f_B\|_1\longrightarrow0,
 \qquad
 \|a_{B,h}-a_B\|_1\longrightarrow0.
\end{equation}

Extend all four densities by zero to a fixed compact rectangle and define
\[
 w_h(s,b):=f_{S,h}(s)a_{B,h}(b)
          -a_{S,h}(s)f_{B,h}(b),
 \qquad
 w(s,b):=f(s)a_B(b)-a_S(s)f_B(b).
\]
The product structure and
\cref{eq:seller-L1,eq:buyer-L1} imply
\begin{equation}\label{eq:objective-density-L1}
 \|w_h-w\|_1\longrightarrow0.
\end{equation}
For example,
\[
 \|f_{S,h}a_{B,h}-fa_B\|_{L^1(ds\,db)}
 \le \|f_{S,h}\|_1\|a_{B,h}-a_B\|_1
    +\|f_{S,h}-f\|_1\|a_B\|_1.
\]
The second product is bounded in the same way.  This $L^1$ argument
allows the density to be unbounded near the left endpoint.
Let $L_{\alpha,h}$ denote the original-scale version of
\cref{eq:lattice-L} for these grids, and set
$V_{\alpha,h}:=\sup_{x\in\Rtwo}L_{\alpha,h}(x)$.
If $x_{ij}$ is a monotone lattice allocation and $x_h$ is its constant
extension to $I_i\times J_j$, discrete summation by parts gives
\begin{align}
 L_{\alpha,h}(x)
 &=\sum_{i,j}\left[
    g_iA^B_{j,h}
    -p_{j,h}\bigl((1+\alpha)s_i g_i+\alpha hG_{i-1}\bigr)
    \right]x_{ij}\notag\\
 &=\iint w_h(s,b)x_h(s,b)\,ds\,db.
 \label{eq:spread-objective}
\end{align}
Every such extension is pointwise monotone.  Because $0\le x_h\le1$,
\cref{eq:objective-density-L1} yields
\[
 \limsup_{h\downarrow0}V_{\alpha,h}
 \le \mathcal V_\alpha(\widehat\psi\circ c_F,F_B).
\]
Here the right side is written in quantile coordinates; the change of variables
$u=F(s)$ identifies it with the supremum of $\iint w(s,b)x(s,b)\,ds\,db$ over
value-coordinate pointwise-monotone allocations.  Conversely, average any such
allocation on each $I_i\times J_j$.  The cell averages form a monotone lattice
array, and their step extensions converge at almost every Lebesgue point.
Dominated convergence against $|w|$, together with
\cref{eq:objective-density-L1}, gives the reverse inequality.  Hence
\begin{equation}\label{eq:V-lattice-limit}
 V_{\alpha,h}\longrightarrow
 \mathcal V_\alpha(\widehat\psi\circ c_F,F_B).
\end{equation}

Let $c_h$ be the lattice seller's quantile function and set
$\FB_h:=\E[(B_h-S_h)_+]$.  The lattice seller CDFs converge to $F$ by
\cref{eq:lattice-product}: convergence holds uniformly away from zero, and the
mass near zero is controlled by the preceding $F(\delta)+o(1)$ bound.  Since
$F$ is continuous and all distributions have common compact support, their
quantile functions converge in $L^1(0,1)$.  Coupling by quantiles and using
$|B_h-B|\le h$ gives
\begin{equation}\label{eq:FB-lattice-limit}
 |\FB_h-\FB|
 \le h+\int_0^1|c_h(u)-c_F(u)|\,du\longrightarrow0.
\end{equation}
Applying \cref{cor:lattice-bound} before taking the limit proves the claim for
smooth strictly regular sellers.

For the general $\psi$ in the statement, extend it monotonically to the real
line by setting it equal to zero to the left of $0$ and equal to its terminal
value to the right of $1$.  Convolve with a nonnegative smooth mollifier of
radius $\varepsilon_k\downarrow0$, subtract the value at zero, and add
$\varepsilon_k u$.  The resulting bounded functions $\psi_k$ are smooth and
nonnegative on $[0,1]$, satisfy
\[
 \psi_k(0)=0,
 \qquad \psi_k'(u)\ge\varepsilon_k>0,
 \qquad \|\psi_k-\psi\|_1\longrightarrow0.
\]
Put
$c_k=T_\alpha\psi_k$.  Differentiating \cref{eq:T-alpha} gives
\[
 (1+\alpha)c_k(u)+\alpha u c_k'(u)=\psi_k(u).
\]
The term $(1+\alpha)c_k(u)$ is a weighted average of the past values of
$\psi_k$.  Writing $r=(1+\alpha)/\alpha$, the derivative lower bound gives
\[
 \psi_k(u)-(1+\alpha)c_k(u)
 \ge \varepsilon_k\!
 \left(u-r u^{-r}\int_0^u v^r\,dv\right)
 =\frac{\alpha\varepsilon_k}{1+2\alpha}\,u.
\]
Consequently
\begin{equation}\label{eq:ck-derivative-lower}
 c_k'(u)\ge\frac{\varepsilon_k}{1+2\alpha}>0.
\end{equation}
Smoothness of $\psi_k$ and $\psi_k(0)=0$ also make $c_k$ smooth at zero.
Thus $c_k$ is the quantile function of a seller with a strictly positive
smooth density, and its value-space $\alpha$-virtual cost is strictly
increasing.

The operator $T_\alpha$ is an $L^1$ contraction.  For
$\delta_k=\psi_k-\psi$, Fubini's theorem gives
\begin{align}
 \|c_k-c\|_1
 &\le\frac1\alpha\int_0^1v^{1/\alpha}|\delta_k(v)|
       \int_v^1u^{-1-1/\alpha}\,du\,dv\notag\\
 &=\int_0^1(1-v^{1/\alpha})|\delta_k(v)|\,dv
 \le\|\delta_k\|_1.
 \label{eq:T-contraction}
\end{align}
Also, uniformly over $X\in\mathcal X$,
\begin{equation}\label{eq:V-L1-continuity}
 |\mathcal L_\alpha(\psi_k;X)-\mathcal L_\alpha(\psi;X)|
 \le\|\psi_k-\psi\|_1.
\end{equation}
Thus the corresponding quantile values converge.  The positive-part map is
one-Lipschitz, so
\[
 |\FB(c_k,F_B)-\FB(c,F_B)|\le\|c_k-c\|_1\longrightarrow0.
\]

If $c$ is constant on a nondegenerate quantile interval $I$, its defining ODE
gives $\psi=(1+\alpha)c$ almost everywhere on $I$.  Averaging the allocation
rows over $I$ preserves monotonicity and both terms in
\cref{eq:quantile-objective}.  Hence the quantile formulation has the same
value when all rows corresponding to one seller atom are identified.  Apply
the smooth strictly regular case to $c_k$ and pass to the limit.
\end{proof}

\begin{proof}[Proof of \cref{thm:fixed-alpha}]
Apply \cref{lem:fixed-ironing} to the seller quantile $c$ and then
\cref{lem:regular-lattice-limit} to $\overline\psi$.  Since the original
seller quantile is strictly increasing,
\[
 V_\alpha
 =\mathcal V_\alpha(\psi_\alpha,F_B)
 =\mathcal V_\alpha(\overline\psi,F_B)
 \ge\beta(\alpha)\FB(\overline c,F_B)
 \ge\beta(\alpha)\FB(c,F_B).
\]
At $\alpha=1$, $q=1/2$ and
\[
 C(1)=\frac12\BetaFn\!\left(\frac12,\frac12\right)=\frac\pi2.
\]
\end{proof}

\section{From the fixed multiplier to \texorpdfstring{$1/\pi$}{1/pi}}\label{sec:pi}

\begin{lemma}[Continuous virtual threshold identities]\label{lem:threshold-id}
For smooth priors and every price $p$,
\begin{align}
 \int_p^1(\phi_B(b)-s)f_B(b)\,db
   &=(p-s)\Prb[B\ge p], \label{eq:buyer-threshold-id}\\
 \int_0^p(b-\psi_S(s))f_S(s)\,ds
   &=(b-p)\Prb[S\le p]. \label{eq:seller-threshold-id}
\end{align}
\end{lemma}

\begin{proof}
Substitute \cref{eq:ordinary-virtual} and integrate the tail term in the first
identity and the CDF term in the second.
\end{proof}

\begin{lemma}[Lagrangian value is bounded by proposer profits]\label{lem:V-profit}
For smooth positive-density priors,
\[
 V_1\le\Pi_S+\Pi_B.
\]
\end{lemma}

\begin{proof}
For $x\in\Rtwo$, multiplier one gives
\begin{align}
 \GFT(x)+\Lambda(x)
 &=\E[(\phi_B(B)-S)x(S,B)]
   +\E[(B-\psi_S(S))x(S,B)].
 \label{eq:V-decomp}
\end{align}
Fix $s$ and choose a right-continuous version of the nondecreasing row
$b\mapsto x(s,b)$.  Its Stieltjes measure $\mu_s$ has total mass at most one
and satisfies, almost everywhere,
\[
 x(s,b)=\int_{[0,1]}\mathbf1\{b\ge p\}\,d\mu_s(p);
\]
the missing mass is the no-trade rule.  Fubini's theorem and
\cref{eq:buyer-threshold-id} express the row contribution to the first term of
\cref{eq:V-decomp} as a mixture of seller posted-price profits.  It is at most
the seller's optimal profit at $s$.  Averaging over $S$ bounds the first term
by $\Pi_S$.

For fixed $b$, the same construction applied to the nonincreasing column gives
a subprobability mixture of rules $\mathbf1\{s\le p\}$.  By
\cref{eq:seller-threshold-id}, its contribution is at most the buyer's optimal
profit at $b$.  Averaging over $B$ and taking the supremum over $x$ proves the
lemma.
\end{proof}

\begin{lemma}[Smooth approximation of Borel priors]\label{lem:smooth-approx}
For any independent Borel priors on $[0,1]$, there are independent priors with
strictly positive smooth densities such that
\[
 \FB_k\to\FB,
 \qquad \Pi_{S,k}\to\Pi_S,
 \qquad \Pi_{B,k}\to\Pi_B.
\]
\end{lemma}

\begin{proof}
We first establish stability under an $L^\infty$ coupling.  Suppose
$|S-S'|\le\varepsilon$ and $|B-B'|\le\varepsilon$ almost surely.  Since the
positive-part map is one-Lipschitz in each argument,
\[
 |\E(B-S)_+-\E(B'-S')_+|\le2\varepsilon.
\]

Consider seller profit.  Fix paired seller types $s,s'$ and a price $p$ for
the original instance.  If $p-s\le2\varepsilon$, its profit is at most
$2\varepsilon$, so the no-trade option in the perturbed instance gives the
required comparison.  If $p-s>2\varepsilon$, post
$p^-=(p-\varepsilon)\vee0$.  Then
\[
 \{B\ge p\}\subseteq\{B'\ge p^-\},
 \qquad p^- -s'\ge p-s-2\varepsilon>0.
\]
The profit from $p^-$ in the perturbed instance is therefore at least the
profit from $p$ minus $2\varepsilon$.  Taking the supremum over $p$ and then
interchanging the primed and unprimed instances shows that the conditional,
and hence ex-ante, seller-profit values differ by at most $2\varepsilon$.
For buyer profit, the symmetric transfer is
$p^+=(p+\varepsilon)\wedge1$; the inclusion
$\{S\le p\}\subseteq\{S'\le p^+\}$ gives the same bound.

Choose $\varepsilon_k\downarrow0$.  Map each original value $v\in[0,1]$ to
$2\varepsilon_k+(1-4\varepsilon_k)v$ and add independent noise with a smooth
density supported on $[-\varepsilon_k,\varepsilon_k]$.  The resulting prior
has a smooth density supported in $[\varepsilon_k,1-\varepsilon_k]$ and can be
coupled to the original value with displacement at most
$3\varepsilon_k$.  Mix it with the uniform distribution on $[0,1]$ with
probability $\delta_k\downarrow0$.  The mixture has a smooth density that is
strictly positive on $[0,1]$.

Mixing changes each prior by total variation at most $\delta_k$.  Conditional
proposer profit and the first-best integrand lie in $[0,1]$; responder tails
and CDFs change by at most $\delta_k$.  Thus changing both priors by these
mixtures changes each proposer-profit value and first best by at most
$2\delta_k$.  Combining this estimate with the coupling bounds above proves
all three convergences.
\end{proof}

\begin{proof}[Proof of \cref{thm:pi-main}]
For smooth positive-density priors, \cref{thm:fixed-alpha} at $\alpha=1$ and
\cref{lem:V-profit,lem:gft-profit} give
\[
 \frac2\pi\FB\le V_1\le\Pi_S+\Pi_B\le\SO+\BO=2\RO.
\]
For arbitrary Borel priors, apply the smooth inequality to the approximating
sequence in \cref{lem:smooth-approx} and pass to the limit:
\[
 \frac2\pi\FB\le\Pi_S+\Pi_B\le2\RO.
\]
Dividing by two proves the theorem.
\end{proof}

Fei's constant is $1/T$, where $T>1$ is the larger solution of
$T=2+\log T$, and
\[
 \frac1T=0.317844432899372\ldots
 <\frac1\pi=0.318309886183790\ldots.
\]
The value $\pi$ is the beta integral at the multiplier where the two
proposer-profit objectives enter symmetrically.

\section{An explicit hard family}\label{sec:hard-instance}

We now prove \cref{thm:upper-main}.  The family depends on two vanishing
parameters: the buyer truncation $\varepsilon>0$ and a seller perturbation
$\eta>0$.  We first define the limiting shape.

Fix parameters $a,k_1,k_2,x,z>0$ with $0<x<z<1$, and define
\begin{equation}\label{eq:Q-A}
 Q(y):=y^a\exp\!\bigl(k_1(y-1)+k_2(y^2-1)\bigr),
 \qquad
 A(y):=a+k_1y+2k_2y^2,
 \quad 0\le y\le1.
\end{equation}
Then $Q(1)=1$ and $Q'(y)=Q(y)A(y)/y$.

\subsection{The two distributions}

\begin{definition}[Truncated equal-revenue buyer]\label{def:equal-revenue}
For $\varepsilon\in(0,1)$, let $B_\varepsilon$ be the buyer value with
survival function
\[
 \Prb[B_\varepsilon\ge b]=
 \begin{cases}
  1, & 0\le b\le\varepsilon,\\
  \varepsilon/b, & \varepsilon<b\le1.
 \end{cases}
\]
Equivalently, the distribution has density $\varepsilon/b^2$ on
$[\varepsilon,1)$ and an atom of mass $\varepsilon$ at $1$.
\end{definition}

\begin{definition}[Three-region seller]\label{def:seller-hard}
For $\eta>0$, set
\[
 q_\eta:=\frac{1+\eta-z}{1+\eta-x}.
\]
Let $S_\eta$ be the seller value with CDF
\begin{equation}\label{eq:Geta}
 G_\eta(s):=
 \begin{cases}
  q_\eta Q(s/x), & 0\le s\le x,\\[2mm]
  \dfrac{1+\eta-z}{1+\eta-s}, & x<s<z,\\[3mm]
  1, & z\le s\le1.
 \end{cases}
\end{equation}
\end{definition}

\subsection{Distribution validity and optimal offers}

For $s=xy\in(0,x]$, define
\begin{equation}\label{eq:psi-lower}
 \psi(s):=xy\frac{1+A(y)}{A(y)}.
\end{equation}

\begin{lemma}[Validity and optimal offers]\label{lem:hard-validity}
Assume $a>2k_2$ and $\psi(x)<1$.  Then $G_\eta$ is a continuous CDF.  The
seller-offering mechanism posts $1$ almost surely, and the buyer's unique optimal offer is
\begin{equation}\label{eq:buyer-offer-rule}
 p(b)=
 \begin{cases}
  \psi^{-1}(b), & 0<b\le\psi(x),\\
  x, & \psi(x)<b\le1.
 \end{cases}
\end{equation}
\end{lemma}

\begin{proof}
The pieces agree at $x$ and $z$.  Moreover, $Q(0)=0$, $Q(1)=1$, and
$Q'(y)=Q(y)A(y)/y>0$, so $G_\eta$ is a continuous CDF.  Let $g_\eta$ denote
its density on $(0,z)$.

For every seller value $s>0$, a buyer price $p\in(\varepsilon,1]$ gives seller
profit
\[
   (p-s)\frac{\varepsilon}{p}
   =\varepsilon\left(1-\frac{s}{p}\right),
\]
which is increasing in $p$.  Prices below $\varepsilon$ give weakly smaller
profit: their best value is attained at $p=\varepsilon$, and
$\varepsilon(1-s)\ge\varepsilon-s$.  Thus the seller posts $p=1$, uniquely
for $s>0$.  Since $G_\eta$ has zero mass at $s=0$,
\begin{equation}\label{eq:SO-limit}
   \frac{\SO_{\varepsilon,\eta}}{\varepsilon}
   =\E[1-S_\eta]
   =\int_0^1G_\eta(s)\,ds.
\end{equation}

For the buyer's problem, define the seller virtual cost
\[
   \psi_\eta(s):=s+\frac{G_\eta(s)}{g_\eta(s)}.
\]
On the lower region, with $y=s/x$, one has $\psi_\eta(s)=\psi(s)$.
Its derivative is strictly positive:
\begin{equation}\label{eq:psi-derivative}
 \frac1x\frac{d\psi}{dy}
 =1+\frac{A(y)-yA'(y)}{A(y)^2}
 =1+\frac{a-2k_2y^2}{A(y)^2}>0.
\end{equation}
On the middle region, direct differentiation gives
\begin{equation}\label{eq:psi-cap}
    \psi_\eta(s)=1+\eta.
\end{equation}

On each region where $G_\eta$ has a density,
\[
 \frac{d}{dp}\bigl((b-p)G_\eta(p)\bigr)
 =g_\eta(p)\bigl(b-\psi_\eta(p)\bigr).
\]
The lower-region virtual cost is continuous and strictly increasing from $0$
to $\psi(x)$.  Thus, for $b\le\psi(x)$, profit increases up to
$p=\psi^{-1}(b)$ and decreases afterward.  For $b>\psi(x)$, profit increases
throughout the lower region and is maximized there at $x$.  On the middle
region, $\psi_\eta=1+\eta>b$, so profit is strictly decreasing; on
$[z,1]$, it equals $b-p$ and is also decreasing.  This proves
\cref{eq:buyer-offer-rule} and uniqueness.
\end{proof}

\subsection{The limiting ratio}

Let $G$ be \cref{eq:Geta} at $\eta=0$, and set
\begin{equation}\label{eq:F-S-B-integrals}
 \mathcal F:=\int_0^1\frac{G(s)}s\,ds,
 \qquad
 \mathcal S:=\int_0^1G(s)\,ds,
 \qquad
 \mathcal B:=\int_0^x
 \left(1-\frac{s}{\psi(s)}-\log\psi(s)\right)dG(s).
\end{equation}

\begin{lemma}[Sequential limit]\label{lem:sequential-limit}
For the distributions in \cref{def:equal-revenue,def:seller-hard},
\[
 \lim_{\eta\downarrow0}\lim_{\varepsilon\downarrow0}
 \frac{\FB_{\varepsilon,\eta}}{\RO_{\varepsilon,\eta}}
 =\frac{2\mathcal F}{\mathcal S+\mathcal B}.
\]
\end{lemma}

\begin{proof}
The first-best identity
\[
 \FB_{\varepsilon,\eta}
 =\int_0^1G_\eta(s)\Prb[B_\varepsilon\ge s] \,ds
\]
gives
\[
 \lim_{\varepsilon\downarrow0}
 \frac{\FB_{\varepsilon,\eta}}\varepsilon
 =\int_0^1\frac{G_\eta(s)}s\,ds.
\]
Near zero, $G_\eta(s)$ is a constant multiple of $s^a$, so
$G_\eta(s)/s$ is integrable.  Moreover,
$\Prb[B_\varepsilon\ge s]/\varepsilon\le1/s$ for $s>0$.  Dominated
convergence therefore proves the limit.  \Cref{eq:SO-limit} gives the
seller-offering limit.

Under \cref{eq:buyer-offer-rule}, a seller type $s\le x$ trades precisely when
$B_\varepsilon\ge\psi(s)$.  For every fixed $v\in(0,1)$,
\begin{align*}
 \lim_{\varepsilon\downarrow0}
 \frac1\varepsilon\E[(B_\varepsilon-s)
       \mathbf1\{B_\varepsilon\ge v\}]
 &=\int_v^1\frac{b-s}{b^2}\,db+(1-s)\\
 &=1-\frac{s}{v}-\log v.
\end{align*}
Substitute $v=\psi(s)$ and integrate over $dG_\eta(s)$ on $[0,x]$.
Seller types above $x$ never trade under a buyer offer.  Since $v=\psi(s)>s$, the scaled integrand has the common bound
\[
 \frac1\varepsilon\E[(B_\varepsilon-s)
       \mathbf1\{B_\varepsilon\ge v\}]
 \le 1+|\log v|.
\]
Indeed, the displayed limit formula gives the bound when $\varepsilon\le v$;
when $v<\varepsilon$, the left side equals
$1-\log\varepsilon-s/\varepsilon\le1-\log v$.
\begin{samepage}
As $s\downarrow0$, $\psi(s)\asymp s$ and
$dG_\eta(s)=O(s^{a-1})\,ds$, uniformly for small $\eta$.  Hence
$s^{a-1}(1+|\log s|)$ is a common integrable majorant, and the buyer-offering
integral converges to $\mathcal B$.  Also, $G_\eta\to G$
pointwise, while $G_\eta(s)/s$ is bounded by a constant multiple of
$s^{a-1}$ on $(0,x]$ and uniformly on $[x,1]$.  Dominated convergence gives
\[
 \int_0^1\frac{G_\eta(s)}s\,ds\longrightarrow\mathcal F,
 \qquad
 \int_0^1G_\eta(s)\,ds\longrightarrow\mathcal S.
\]
\end{samepage}
Combining the first-best, seller-offering, and buyer-offering limits and using
$\RO=(\SO+\BO)/2$ gives the stated ratio.  Taking the limits in this order
keeps the buyer's boundary offer unique during the $\varepsilon$ limit.
\end{proof}

The three quantities have convenient unit-interval formulas.  Put
\begin{equation}\label{eq:I-defs}
 I_F:=\int_0^1\frac{Q(y)}y\,dy,
 \qquad
 I_S:=\int_0^1Q(y)\,dy,
\end{equation}
and
\begin{equation}\label{eq:IB-def}
 I_B:=\int_0^1
 \left[
 \frac1{1+A(y)}
 -\log\!\left(xy\frac{1+A(y)}{A(y)}\right)
 \right]Q'(y)\,dy.
\end{equation}
With $q_0=(1-z)/(1-x)$, elementary integration on the middle and terminal
regions gives
\begin{align}
 \mathcal F
 &=q_0I_F
 +(1-z)\log\!\frac{z(1-x)}{x(1-z)}-\log z,
 \label{eq:F-explicit}\\
 \mathcal S
 &=q_0xI_S
 +(1-z)\log\!\frac{1-x}{1-z}+(1-z),
 \label{eq:S-explicit}\\
 \mathcal B&=q_0I_B.
 \label{eq:B-explicit}
\end{align}

\subsection{Certified parameters}

Take the following terminating decimals as exact rational parameters:
\begin{align}
 a&=0.095412847191313230,\notag\\
 k_1&=0.0034331918146865968,\notag\\
 k_2&=0.000042063244380725156,\label{eq:parameters}\\
 x&=0.027131876557804906,\notag\\
 z&=0.84530040073301205.\notag
\end{align}
They satisfy $a>2k_2$ and
\[
 A(1)=0.098930165494761277112,
 \qquad
 \psi(x)=0.3013846934020427791\ldots<1,
\]
which verifies the conditions used in \cref{eq:psi-derivative,eq:buyer-offer-rule}.

Arb interval arithmetic at 100 decimal digits gives the enclosures in
\cref{tab:certificate}; every displayed radius has been rounded upward.

\begin{table}[t]
\centering
\caption{Certified limiting quantities for \cref{eq:parameters}.}
\label{tab:certificate}
\begin{tabular}{@{}lc@{}}
\toprule
Quantity & Certified enclosure midpoint and radius \\
\midrule
$I_F$ & $10.44755489238983467309805621763\ \pm 10^{-28}$ \\
$I_S$ & $0.91137899753252926492220359010\ \pm 10^{-29}$ \\
$I_B$ & $12.52948706468993154373219525462\ \pm 10^{-28}$ \\
$\mathcal F$ & $2.645835354262006463976354346239\ \pm 10^{-29}$ \\
$\mathcal S$ & $0.443087560739101222493372791649\ \pm 10^{-30}$ \\
$\mathcal B$ & $1.992363179780563445946824002812\ \pm 10^{-29}$ \\
$2\mathcal F/(\mathcal S+\mathcal B)$
 & $2.172768523084520237024664850580\ \pm 10^{-29}$ \\
\bottomrule
\end{tabular}
\end{table}

In particular,
\[
    \frac{2\mathcal F}{\mathcal S+\mathcal B}
    >2.17276852308451.
\]
By \cref{lem:sequential-limit}, sufficiently small positive $\eta$ followed by
sufficiently small positive $\varepsilon$ yields an actual pair of
distributions with the same strict inequality.  This proves
\cref{thm:upper-main}.  The interval-arithmetic verification is given in
\cref{app:certificate}.

\section{Conclusion}

We proved
\[
   \frac1\pi\le\rho_{\rm RO}<0.460242308085529.
\]
The lower bound follows from the fixed-multiplier guarantee at $\alpha=1$,
where the Beta-function constant is $\pi/2$ and the Lagrangian decomposes into
seller and buyer posted-price profit terms.  The upper bound follows from an
explicit hard family with certified limiting ratio
$\FB/\RO>2.17276852308451$.

Determining the exact value of $\rho_{\rm RO}$ remains open.  The lower bound
is governed by a multiplier-dependent monotone-allocation argument, whereas
the hard family is governed by the seller's lower-tail virtual-cost curve.  A
tight characterization must connect these two structures.

\appendix

\section{Proof of the parameterized greedy bound}\label{app:parameter-audit}

This appendix proves \cref{prop:capacity} by extending the state argument of
Liu et al.~\cite[Sections~4--5 and Appendix~C]{LQRW26} from their target
$\beta=1/2$ to a symbolic target.  Their maximin--greedy equivalence already
holds for arbitrary $\beta$; the work here is to verify the no-failure lemmas
with the target left symbolic.  Fix $\alpha>0$, let $q=1/(1+\alpha)$ and
$a=1-q$, and assume
\begin{equation}\label{eq:beta-condition-app}
 0<\beta C(\alpha)\le1,
 \qquad C(\alpha)=q\BetaFn(q,q).
\end{equation}
Since
\[
 C(\alpha)=\int_0^1q t^{q-1}(1-t)^{q-1}\,dt>1,
\]
the condition also gives $0<\beta<1$, as required by the state recursion.

For a full-support lattice seller, set
\[
 r_i:=\frac{F_S(i-1)}{f_S(i)},\qquad i\ge1,
\]
and set $r_0:=r_{\rm base}:=0$.  Thus
$\psi_S^\alpha(i)=(1+\alpha)i+\alpha r_i$ for $i\ge1$.  At a successful tight
step $H(b;x)=0$, the current column has threshold form
\[
 x(s,b)=1\ (s<t),\qquad x(t,b)=y\in(0,1],\qquad x(s,b)=0\ (s>t).
\]
Its \emph{clear state} is $\omega=(y;r_t,\ldots,r_{b-1})$.  The base clear
state is $\omega_0=(\beta q;r_{\rm base})$.  Given a clear state reached at
buyer value $b_0$, the compressed execution evaluates the positive
normalization
\[
 \widehat H_{t,b_0}(b;x)
 :=\frac{H(b;x)-H(b_0;x)}{F_S(t)},\qquad b>b_0.
\]
The clear-state coordinates and the future ratio tail determine every later
greedy update.  Write $\operatorname{ALG}(\omega;\mathbf r)$ for this
compressed execution with future ratios $\mathbf r$.

Two clear states are comparable when one ratio vector is a suffix of the
other.  Write $\omega\prec\omega'$ when the ratio vector of $\omega'$ is a
proper suffix of that of $\omega$, or when the vectors agree and the threshold
allocation of $\omega'$ is larger.  Thus $\omega'$ is the more demanding
state.  The following lemma records the two state properties used below.

\begin{lemma}[Greedy-state toolkit]\label{lem:imported-toolkit}
Fix $\alpha\ge0$, $\beta\ge0$, and an $\alpha$-weakly regular full-support
lattice seller.
\begin{enumerate}[label=(\roman*)]
 \item For a fixed terminal value $M$, the maximin inequality
 \cref{eq:fixed-maximin} holds exactly when the greedy construction succeeds
 through $M$.  The inequalities for all finite $M$ are therefore equivalent
 to never failing.
 \item Clear states are sufficient statistics.  If $\omega\prec\omega'$ and
 $\operatorname{ALG}(\omega';\mathbf r)$ succeeds, then
 $\operatorname{ALG}(\omega;\mathbf r)$ succeeds and its output is dominated
 by the output from $\omega'$.  Equivalently, failure from $\omega$ implies
 failure from $\omega'$.
\end{enumerate}
\end{lemma}

\begin{proof}
Part~(i) is Liu et al.'s Maximin--Greedy Equivalence
\cite[Theorem~4.6]{LQRW26}, stated there for arbitrary $\beta$.

For part~(ii), subtract the last clear equality $H(b_0;x)=0$ from each later
constraint and divide by the current threshold CDF.  This is the normalized
state equation used in their Proposition~5.1.  Its coefficients depend only on
the clear state and the future ratio tail.  In the comparison of two aligned
states, the term proportional to $\beta$ is identical in both normalized
equations and cancels.  Suppose their threshold order reverses for the first
time at buyer value $b$.  Earlier columns are ordered by the induction
hypothesis, and their coefficients in the normalized equation are
nonpositive; the current-column coefficients have the opposite strict order.
The two equations therefore satisfy
\[
 \widehat H(b;x')<\widehat H(b;x)=0
\]
for the more demanding run, contradicting its clear equality.  Successful
outputs preserve the dominance order.

For failure propagation, the ordered thresholds and weak regularity imply that
a virtual-cost-test failure in the less demanding run also occurs in the more
demanding run.  The other failure mode is exhaustion of the boundary
candidates.  A feasible fractional boundary value in the more demanding run,
together with the same normalized comparison, would give one in the less
demanding run.  Thus failure from the less demanding state implies failure from
the more demanding state.  These are the remaining steps of Liu et
al.~\cite[Proposition~5.1]{LQRW26}; after the target term cancels, they do not
depend on the value of $\beta$.
\end{proof}

An execution is \emph{threshold-stationary} through length $N$ when it clears
at threshold zero for every buyer value below $N$.  A \emph{proper-suffix
output} has deleted at least one positive leading ratio.
For a ratio sequence, normalize $F_N(0)=1$ and set
\[
 F_N(m)=\prod_{i=1}^{m}\left(1+\frac1{r_i}\right).
\]
The long execution uses the sequence
$(r_{\rm base},r_1,\ldots,r_{k+1})$.  The short execution deletes $r_1$ and
uses $(r_{\rm base},r_2,\ldots,r_{k+1})$.  The quantities
$\mathcal D_\beta$ and $Q_b$ below are the corresponding normalized gaps
between the long state and the artificial short state.

The next three subsections establish the parameterized versions of their
threshold-one feasibility, threshold-stationary truncation, and nonstationary
truncation results.  The final subsection applies their minimal-counterexample
argument.

\subsection{Feasibility at threshold one}

While the threshold remains zero, successive differences of the tight
constraints $H(b;x)=0$ give
\begin{equation}\label{eq:G-recurrence-app}
 x_b=\beta G_b,
 \qquad
 G_b=\left(1-\frac{q}{b}\right)G_{b-1}
 +\frac{q}{b}\frac{F_S(b-1)}{F_S(0)},
 \qquad G_1=q.
\end{equation}
Unrolling the recurrence gives
\[
 G_b=\sum_{t=1}^b\frac qt\frac{F_S(t-1)}{F_S(0)}
       \prod_{j=t+1}^b\left(1-\frac qj\right).
\]
Since
\[
 \sum_{t=1}^b\frac qt\prod_{j=t+1}^b\left(1-\frac qj\right)
 =1-\prod_{j=1}^b\left(1-\frac qj\right)<1,
\]
monotonicity of $F_S$ gives $G_b<F_S(b)/F_S(0)$.  Subtracting consecutive
instances of \cref{eq:G-recurrence-app} now yields
\[
 G_{b+1}-G_b=\frac{q}{b+1}
 \left(\frac{F_S(b)}{F_S(0)}-G_b\right)>0.
\]

Suppose the virtual-cost test failed at threshold one for buyer value $b$.
Put $\kappa=(1+\alpha)/\alpha=1/a$ and normalize $F_S(0)=1$.  Weak regularity
would imply, for $1\le m\le b-1$,
\[
 r_m=\frac{F_S(m-1)}{f_S(m)}\ge\kappa(b-m).
\]
Consequently, for $0\le i\le b-1$,
\[
 F_S(i)=\prod_{m=1}^i\left(1+\frac1{r_m}\right)
 \le\prod_{m=1}^i\left(1+\frac{a}{b-m}\right)
 =:\widetilde F^{(b+1)}(i).
\]
All coefficients in the unrolled formula for $G_b$ are nonnegative, so
$G_b$ is bounded by the corresponding extremal value
$\widetilde G_b^{(b+1)}$.  The beta--gamma calculation of Liu et
al.~\cite[Lemma~5.5]{LQRW26} is independent of the target and gives
\begin{equation}\label{eq:G-C-app}
 G_b\le\widetilde G_b^{(b+1)}<C(\alpha)=q\BetaFn(q,q).
\end{equation}
The threshold-zero candidate therefore respects the copied-column lower bound,
and $x_b=\beta G_b<\beta C(\alpha)\le1$.  It is feasible and is encountered
before threshold one.  This contradicts the hypothetical failure and proves the
parameterized form of Proposition~5.4.

\subsection{Threshold-stationary truncation}

Consider the long and short executions in Proposition~5.6, of lengths $k+2$
and $k+1$, and let $F_{k+2}$ be the normalized CDF generated by the long
inverse-hazard-ratio sequence.  Put
\[
 \varphi_i:=(i+1)(1+\alpha)+\alpha r_{i+1},
 \qquad i=0,\ldots,k.
\]
Weak regularity is $\varphi_0\le\cdots\le\varphi_k$.  Replace the factors $1/2$ in the two auxiliary
quantities in their Appendix~C by
\begin{align}
 K_\beta
 &:={}
 \sum_{m=0}^{k+1}
 \frac{\beta F_{k+2}(m)}{(1+\alpha)(m+1)}
 \prod_{j=m+2}^{k+2}\left(1-\frac qj\right),
 \label{eq:Kbeta}\\
 J_\beta
 &:={}
 \sum_{m=1}^{k+1}
 \frac{\beta\alpha(k+1-m)F_{k+2}(m)}
      {(1+\alpha)m(m+1)}
 \prod_{j=m+2}^{k+1}\left(1-\frac qj\right).
 \label{eq:Jbeta}
\end{align}
Thus
\begin{equation}\label{eq:KJ-scaling}
 K_\beta=2\beta K_{1/2},
 \qquad J_\beta=2\beta J_{1/2}.
\end{equation}
Here $K_\beta$ is the allocation required to clear the threshold-zero
candidate at buyer value $k+2$.  If $K_\beta<1$, that candidate is feasible;
if $K_\beta=1$, it clears with allocation one.  In either case the output
retains the first positive ratio.  Thus a proper-suffix output implies
$K_\beta>1$, and the weaker inequality $K_\beta\ge1$ will suffice.

For a candidate whose ratio suffix begins at $t\ge2$, let
$x^{\rm long}$ denote the long candidate and let $x^{\rm short}$ be the
corresponding artificial candidate in the short execution, with the same
boundary allocation.  The candidate-dependent terms in the two normalized
constraints cancel, as in equations~(12)--(15) of Liu et al.  Define
$\mathcal D_\beta$ by
\begin{equation}\label{eq:D-residual-app}
 \widehat H^{\rm long}(k+2;x^{\rm long})
 -\left(1+\frac1{r_1}\right)
  \widehat H^{\rm short}(k+1;x^{\rm short})
 =\mathcal D_\beta.
\end{equation}
The right side is independent of $t$ and of the boundary allocation.

The normalized CDFs of the two executions satisfy
$F_{k+2}(m+1)=(1+1/r_1)F_{k+1}(m)$.  Substituting the threshold-zero recurrence
\cref{eq:G-recurrence-app} into the long-minus-short state residual and
collecting the terms gives
\begin{equation}\label{eq:D-beta-app}
 \mathcal D_\beta
 =J_\beta+1-\beta(k+2)
    \prod_{j=2}^{k+2}\left(1-\frac qj\right).
\end{equation}
The two terms after $J_\beta$ are independent of the virtual-cost vector.
Thus, on every level set of $K_\beta$, minimizing $\mathcal D_\beta$ is the
same as minimizing $J_\beta$.  Equation~\eqref{eq:KJ-scaling} already shows
that the level sets and every first-order comparison are exactly those in the
$\beta=1/2$ proof, up to the common positive factor $2\beta$.

The adjacent-pooling calculation is as follows.  Let
\[
 d_\ell:=\varphi_\ell-(\ell+1)(1+\alpha)>0.
\]
For $m\ge\ell+1$,
\begin{equation}\label{eq:F-derivative-app}
 \frac{\partial F_{k+2}(m)}{\partial\varphi_\ell}
 =-\frac{\alpha}{d_\ell(d_\ell+\alpha)}F_{k+2}(m),
\end{equation}
and the derivative is zero for $m\le\ell$.  Hence
$\partial_{\varphi_\ell}K_\beta<0$.  Put
\[
 R_\ell:=
 \frac{\partial_{\varphi_\ell}J_\beta}
      {\partial_{\varphi_\ell}K_\beta}.
\]
The common factor $2\beta$ cancels in $R_\ell$.  After the common negative
factor in \cref{eq:F-derivative-app} is removed, $R_\ell$ is a weighted
average over $m=\ell+1,\ldots,k+1$.  For each summation index, the quotient
between the numerator and denominator coefficients is
\[
 \frac{\alpha(1+\alpha)(k+2)(k+1-m)}
 {m((1+\alpha)(k+2)-1)},
\]
which is strictly decreasing in $m$.  The average defining $R_{\ell+1}$ is
the same weighted tail with its first, largest term removed.  Therefore
$R_\ell>R_{\ell+1}$ for $0\le\ell<k$.

If $\varphi_i<\varphi_{i+1}$, the implicit-function theorem gives a local
level-set direction satisfying
\[
 \frac{d\varphi_{i+1}}{d\varphi_i}
 =-\frac{\partial_{\varphi_i}K_\beta}
         {\partial_{\varphi_{i+1}}K_\beta}<0,
 \qquad
 \frac{dJ_\beta}{d\varphi_i}
 =\partial_{\varphi_i}K_\beta(R_i-R_{i+1})<0.
\]
Hence increasing the smaller adjacent virtual cost and decreasing the larger
one preserves $K_\beta$, preserves weak regularity for a sufficiently small
move, and strictly lowers $J_\beta$.  Repeating these pooling moves forces every minimizer in the compactified
admissible level set to satisfy $\varphi_0=\cdots=\varphi_k$.  The
compactification and continuity argument in Liu et
al.~\cite[Proof of Proposition~5.6]{LQRW26} applies without change: by
\cref{eq:KJ-scaling}, fixing $K_\beta$ is equivalent to fixing $K_{1/2}$,
and $J_\beta$ is the same positive multiple of $J_{1/2}$.

For that constant-virtual-cost case, let $r>0$ be the terminal positive ratio
and write $\lambda:=ar$.  The long and short sequences are, respectively,
\[
 r_i^{(k+2)}=\frac{k+1-i+\lambda}{a}\quad(1\le i\le k+1),
 \qquad
 r_i^{(k+1)}=\frac{k-i+\lambda}{a}\quad(1\le i\le k).
\]
For $t\ge0$, let $(a)_t=a(a+1)\cdots(a+t-1)$, with $(a)_0=1$.  Define
\begin{equation}\label{eq:U-app}
 U_N(\lambda):=
 \sum_{t=1}^{N}\frac qtF_N(t-1)
 \prod_{j=t+1}^{N}\left(1-\frac qj\right).
\end{equation}
Rejection of the threshold-zero candidate gives
$\beta U_{k+2}(\lambda)\ge1$.  The summands are decreasing in $\lambda$, and
the gamma--beta calculation in Lemma~5.8 gives
\begin{equation}\label{eq:U-cutoff-app}
 U_{k+2}(q)
 <\sum_{j=0}^{\infty}\frac{q}{j+q}\frac{(a)_j}{j!}
 =q\BetaFn(q,q)=C(\alpha)\le\frac1\beta.
\end{equation}
It follows that $\lambda<q$.

Set
\[
 L_k(\lambda):=
 \left(1+\frac{a}{k+\lambda}\right)(k+1)U_{k+1}(\lambda)
 -(k+2)U_{k+2}(\lambda).
\]
Substituting the two $U$-expressions into \cref{eq:D-beta-app} and collecting
the common factor $\beta/q$ gives
\begin{equation}\label{eq:D-L-identity}
 \mathcal D_\beta
 =\frac{\beta}{q}\left(\frac q\beta+L_k(\lambda)\right).
\end{equation}
To bound $L_k$, write $U_{k+1}(\lambda)=\sum_{t=1}^{k+1}A_t(\lambda)$, where
\[
 A_t(\lambda):=\frac qtF_{k+1}(t-1)
      \prod_{j=t+1}^{k+1}\left(1-\frac qj\right).
\]
The constant-cost ratios give
\[
 \frac{k+\lambda+1-t}{k+\lambda+1-t+a}A_t
 =\frac{t-1}{t-1+a}A_{t-1}\quad(2\le t\le k+1),
 \qquad
 \left(1+\frac{a}{k+\lambda}\right)F_{k+1}(t-1)=F_{k+2}(t).
\]
Substitution yields $L_k(\lambda)=\sum_{t=1}^{k+1}w_tA_t(\lambda)$ with
\[
 w_1=\frac{ak(k+\lambda+a)}{(1+a)(k+\lambda)}-(k+1+a),
 \qquad
 w_t=\frac{a(k+\lambda+a)(k+1-t)}{(k+\lambda)(t+a)}\quad(t\ge2).
\]
Because $\lambda<q=1-a$, one has $A_1(\lambda)=A_1(q)$,
$A_t(\lambda)\ge A_t(q)$ for $t\ge2$, and
$(k+\lambda+a)/(k+\lambda)\ge(k+1)/(k+1-a)$.  Hence
\begin{align*}
 L_k(\lambda)
 &\ge
 \left(\frac{ak(k+1)}{(1+a)(k+1-a)}-(k+1+a)\right)A_1(q)\\
 &\quad+
 \sum_{t=2}^{k+1}
 \frac{a(k+1)(k+1-t)}{(k+1-a)(t+a)}A_t(q)
 =-qS_k,
\end{align*}
where the recurrence for $A_t(q)$ gives the final telescoping identity and
\[
 S_k:=\frac{(a)_{k+1}}{(k+1)!}
 +\sum_{j=0}^{k}\frac{q}{j+q}\frac{(a)_j}{j!}.
\]
Put $c_j=(a)_j/j!$.  For
$j\ge k+1$,
\[
 \frac{q}{j+q}c_j>\frac{q}{j+1}c_j=c_j-c_{j+1}.
\]
Thus the omitted infinite tail is larger than $c_{k+1}$, the first term in
$S_k$.  Integrating the binomial series
$(1-t)^{-a}=\sum_{j\ge0}c_jt^j$ against $qt^{q-1}$ gives
\begin{equation}\label{eq:finite-beta-sum-app}
 S_k<\sum_{j=0}^{\infty}\frac{q}{j+q}\frac{(a)_j}{j!}
 =q\BetaFn(q,q)=C(\alpha)\le\frac1\beta.
\end{equation}
Combining \cref{eq:D-L-identity,eq:finite-beta-sum-app} gives
\[
 \mathcal D_\beta
 \ge\beta\left(\frac1\beta-S_k\right)>0.
\]
If the long candidate clears, \cref{eq:D-residual-app} and
$\mathcal D_\beta>0$ imply that the corresponding short residual is negative.
The short execution must therefore fail or increase its boundary allocation,
which produces a dominating state.  If the long execution fails a
virtual-cost test, weak regularity transfers the failure to the corresponding
short threshold.  If it exhausts all boundary candidates, the same residual
comparison is negative for every corresponding short candidate, so the short
execution also fails.  This proves the parameterized threshold-stationary
truncation result.

\subsection{Nonstationary propagation}

The nonstationary allocations are affine in $\beta$.  The following
calculation exposes the cancellation in the long--short clear equations.  We use the local indexing of Liu et al.: in this subsection, $r_0>0$
denotes the first positive ratio deleted in the short execution.  At the first turning point $j+2$, let $r_j^*\le r_j$ be the ratio
at which the short execution reaches allocation one.  Their equations~(38)
and~(40), with symbolic $\beta$, are
\begin{align}
 &(1+\alpha)(j+1)(1-x_{j+1}^{(k+1)})
 =\beta\prod_{i=1}^{j-1}\left(1+\frac1{r_i}\right)
       \left(\frac1{r_j^*}-\frac1{r_j}\right),
 \label{eq:initial-short}\\
 &\frac{R_{j+2}}{r_0}(1-x_{j+2}^{(k+2)})
 >\beta\left(1+\frac1{r_0}\right)
   \prod_{i=1}^{j-1}\left(1+\frac1{r_i}\right)
       \left(\frac1{r_j^*}-\frac1{r_j}\right),
 \label{eq:initial-long}
\end{align}
where $R_b:=(b-1)(1+\alpha)-\alpha r_0$.  The long execution passed the
threshold-one virtual-cost test at the turning point, which gives
$R_{j+2}>0$; thereafter $R_b$ increases by $1+\alpha$.  Subtracting the
second multiple of \cref{eq:initial-short} from
\cref{eq:initial-long} gives
\[
 Q_{j+2}:=
 \frac{R_{j+2}}{r_0}(1-x_{j+2}^{(k+2)})
 -\left(1+\frac1{r_0}\right)(1+\alpha)(j+1)
    (1-x_{j+1}^{(k+1)})>0
\]
for every $\beta>0$.

For propagation, write $y_b=x_b^{(k+2)}$,
$z_{b-1}=x_{b-1}^{(k+1)}$, and
$P_b=\prod_{\ell=1}^{b-2}(1+1/r_\ell)$.  When the long execution raises its
allocation and clears the current constraint, the two recurrences are
\begin{align}
 &(b-1)(1+\alpha)(1-z_{b-1})
 =((b-1)(1+\alpha)-1)(1-z_{b-2})+1-\beta P_b,
 \label{eq:short-propagation}\\
 &0=-\frac{R_b}{r_0}(1-y_b)
 +\frac{R_b-1}{r_0}(1-y_{b-1})
 +\left(1+\frac1{r_0}\right)(1-\beta P_b).
 \label{eq:long-propagation}
\end{align}
The identical affine terms $1-\beta P_b$ cancel.  With
\[
 Q_b:=\frac{R_b}{r_0}(1-y_b)
 -\left(1+\frac1{r_0}\right)(b-1)(1+\alpha)(1-z_{b-1}),
\]
subtraction yields
\begin{equation}\label{eq:Q-propagation}
 Q_b=Q_{b-1}+\alpha\left[
 \frac{1-y_{b-1}}{r_0}
 -\left(1+\frac1{r_0}\right)(1-z_{b-2})
 \right].
\end{equation}
The definition of $Q_{b-1}$ gives
\[
 R_{b-1}\left[
 \frac{1-y_{b-1}}{r_0}
 -\left(1+\frac1{r_0}\right)(1-z_{b-2})
 \right]
 =Q_{b-1}+\alpha(r_0+1)(1-z_{b-2})>0.
\]
Thus $Q_b>Q_{b-1}>0$.

For a copied long column, $y_b=y_{b-1}$.  The short threshold-zero
allocations are nondecreasing, so $1-z_{b-1}\le1-z_{b-2}$.  Moreover,
\[
 \frac{R_b}{(b-1)(1+\alpha)}
 =1-\frac{\alpha r_0}{(b-1)(1+\alpha)}
\]
is increasing in $b$.  The inequality $Q_{b-1}>0$ therefore gives $Q_b>0$
directly from the definition of $Q_b$.  At the terminal step,
let $\widehat H_{\rm long}$ and $\widehat H_{\rm short}$ denote the two
constraint residuals after the positive normalization used in the definition
of $Q_b$.  The same subtraction gives
\[
 \widehat H_{\rm long}
 -\left(1+\frac1{r_0}\right)\widehat H_{\rm short}
 =Q_{k+2}>0.
\]
If the long execution clears at its terminal candidate, the artificial short
candidate has negative residual and the short execution must either fail or
raise its boundary allocation, producing a dominating state.  If the long
execution fails, the same strict comparison transfers failure to the short
execution.  This proves both conclusions of the parameterized
nonstationary-truncation proposition.

\subsection{Minimal counterexample}

Assume that the greedy construction fails, and choose a failing global ratio
vector
\[
 (r_0,r_1,\ldots,r_m),\qquad r_0=0,\quad r_i>0\ (i\ge1),
\]
of minimum length.  Failure occurs after the full positive tail has been
processed; otherwise the prefix ending at the first failure is a shorter
counterexample.

First suppose that no successful clear step before failure has positive
threshold.  The execution is threshold-stationary.  Delete $r_1$ and apply the
threshold-stationary truncation proved above.  Failure transfers to
$(r_0,r_2,\ldots,r_m)$, contradicting minimality.

Now suppose a first positive-threshold clear step exists.  If its threshold is
at least two, call its buyer value $k+2$.  The long prefix
$(r_0,\ldots,r_{k+1})$ is threshold-stationary before this step, and its
output has a proper ratio suffix.  Threshold-stationary truncation says that
the prefix with $r_1$ deleted either fails or produces a state dominating the
long-prefix state.  Failure is already a shorter counterexample.  In the
dominance case, append the common remaining tail
$(r_{k+2},\ldots,r_m)$.  State sufficiency identifies the continuation from
the long-prefix state with the original failing execution, and
\cref{lem:imported-toolkit}(ii) transfers failure to the dominating truncated
state.  This again yields a shorter failing vector.

It remains that the first positive threshold equals one.  Let $k+2$ be the
first later buyer value at which the execution either fails or clears at a
threshold at least two.  Such a value exists because the full execution
fails.  Every earlier threshold in this prefix is at most one.  The
nonstationary truncation proved above says that deleting $r_1$ either makes
the short prefix fail or makes its output dominate the long-prefix output.
The first alternative is a shorter counterexample.  In the second, append the
common remaining tail and invoke state sufficiency and
\cref{lem:imported-toolkit}(ii) to transfer the eventual failure to the short
execution.

Every possible first departure from threshold zero produces a shorter failing
ratio vector.  This contradicts minimality.  The greedy construction never
fails, and \cref{lem:imported-toolkit}(i) proves \cref{prop:capacity}.

\section{Interval-arithmetic certificate}\label{app:certificate}

This appendix gives the interval calculation underlying
\cref{tab:certificate}.  The decimal parameters in \cref{eq:parameters} are
interpreted as exact rationals, and all transcendental quantities are enclosed
with Arb ball arithmetic at 100 decimal digits.  Direct interval evaluation
establishes the conditions in \cref{lem:hard-validity}: $0<a<1$,
$0<x<z<1$, $a>2k_2$, $q_0\in(0,1)$, and $\psi(x)<1$.  The remaining
calculation consists of the following finite polynomial and interval
operations.

Write
\[
 h(y)=k_1y+k_2y^2,
 \qquad
 C_0=e^{-(k_1+k_2)},
 \qquad
 Q(y)=C_0y^ae^{h(y)}.
\]
Then
\begin{align}
 I_F&=C_0\int_0^1y^{a-1}e^{h(y)}\,dy,\notag\\
 I_S&=C_0\int_0^1y^ae^{h(y)}\,dy.
\end{align}
For $I_B$, define
\[
 H(u):=\frac{u}{1+u}-u\log x-u\log(1+u)+u\log u.
\]
Using $Q'(y)=C_0y^{a-1}e^{h(y)}A(y)$ and integrating the $\log y$ term in
closed form gives
\begin{equation}\label{eq:IB-series-form}
 I_B=C_0\int_0^1y^{a-1}
 \left[e^{h(y)}H(A(y))-e^{h(y)}A(y)\log y\right]dy.
\end{equation}

Set $L=30$.  Let $E_{L-1}(y)$ be the Taylor polynomial obtained by
truncating the expansion of $e^{h(y)}$ after order $L-1$ in $h(y)$.  Let
$H_{L-1}(y)$ be the Taylor polynomial of $H(a+d)$ through order $L-1$ in $d$,
after substituting $d=A(y)-a$.  Define the finite polynomials
\[
 P(y):=E_{L-1}(y)H_{L-1}(y)=\sum_np_ny^n,
 \qquad
 J(y):=E_{L-1}(y)A(y)=\sum_nj_ny^n.
\]
Their contribution to \cref{eq:IB-series-form} is evaluated as
\[
 C_0\left(\sum_n\frac{p_n}{a+n}
 +\sum_n\frac{j_n}{(a+n)^2}\right).
\]

Set $K=k_1+k_2$ and $D=k_1+2k_2$.  On $[0,1]$, the exponential remainder is
bounded by
\[
 R_E\le e^K\frac{K^L}{L!}.
\]
For $L\ge2$, termwise differentiation of $H$ gives the Taylor remainder bound
\begin{align}
R_H\le{}&
 \frac{D(D/a)^{L-1}}{L(L-1)(1-D/a)}
 +\frac{D(D/(1+a))^{L-1}}
       {L(L-1)(1-D/(1+a))}\notag\\
&+\frac{(D/(1+a))^L}{L(1-D/(1+a))}
 +\frac{(D/(1+a))^L}{(1+a)(1-D/(1+a))}.
\label{eq:RH}
\end{align}
The exact parameters satisfy
\[
 0<a\le A(y)\le A(1)<0.1,
 \qquad -\log x<4,
 \qquad -\log a<3.
\]
Thus $|\log A(y)|\le-\log a<3$, and, for $u=A(y)$,
\[
 |H(u)|<0.1+0.4+0.01+0.3<1.
\]
Hence the uniform remainders for the two
products in \cref{eq:IB-series-form} are at most
\[
 R_P=R_E+e^K R_H,
 \qquad
 R_J=R_E(a+D).
\]
After integration, the error bounds are
\[
 |\Delta I_F|\le C_0R_E/a,
 \quad
 |\Delta I_S|\le C_0R_E/(a+1),
 \quad
 |\Delta I_B|\le C_0(R_P/a+R_J/a^2).
\]
For the exact parameters in \cref{eq:parameters}, outward-rounded Arb
evaluation gives error bounds smaller than $6.71\times10^{-106}$,
$5.84\times10^{-107}$, and $1.19\times10^{-46}$, respectively.  Evaluating
the finite sums and propagating the resulting balls through
\cref{eq:F-explicit,eq:S-explicit,eq:B-explicit} gives the enclosures in
\cref{tab:certificate}.  Directed rounding yields
\[
 \frac{2\mathcal F}{\mathcal S+\mathcal B}
 >2.17276852308451.
\]

\end{document}